\documentclass{article}

     \usepackage[final]{neurips_2019}

\usepackage[utf8]{inputenc} % allow utf-8 input
\usepackage[T1]{fontenc}    % use 8-bit T1 fonts
\usepackage{hyperref}       % hyperlinks
\usepackage{url}            % simple URL typesetting
\usepackage{booktabs}       % professional-quality tables
\usepackage{amsfonts}       % blackboard math symbols
\usepackage{nicefrac}       % compact symbols for 1/2, etc.
\usepackage{microtype}      % microtypography
\usepackage{color}
\usepackage{enumerate}
\usepackage{algorithm}
\usepackage{algorithmic}
\usepackage{bm}
\usepackage{amsmath}
\usepackage{amsthm}

\DeclareFontFamily{OT1}{pzc}{}
\DeclareFontShape{OT1}{pzc}{m}{it}{<-> s * [1.200] pzcmi7t}{}
\DeclareMathAlphabet{\mathpzc}{OT1}{pzc}{m}{it}

\newcommand{\mV}{\mathrm{V}}
\newcommand{\mG}{\mathrm{G}}
\newcommand{\mE}{\mathrm{E}}
\newcommand{\R}{\mathbb{R}}
\newcommand{\Z}{\mathbb{Z}}

\newcommand{\E}{\mathbb{E}}
\newcommand{\betab}{{\bm \beta}}

\newcommand{\yb}{\mathbf{y}}
\newcommand{\Ib}{\mathbf{I}}

\newcommand{\1}{\mathbf{1}}

\newcommand{\Xb}{\mathbf{X}}

\newcommand{\mcD}{\mathcal{D}}

\newcommand{\mN}{\mathrm{N}}
\newcommand{\PPo}{\mathpzc{P}_{\Omega}}
\newcommand{\omegab}{{\bm \omega}}
\newcommand{\betaf}{\betab^\flat}
\newcommand{\betas}{\betab^\sharp}
\newcommand{\MM}{\mathpzc{M}}
\newcommand{\mR}{\mathrm{R}}
\newcommand{\Wb}{\mathbf{W}}
\newcommand{\tX}{\tilde{\Xb}}
\newcommand{\ty}{\tilde{\yb}}
\newcommand{\eb}{\mathbf{e}}
\newcommand{\mO}{\mathcal{O}}

\DeclareMathOperator{\diag}{diag}

\DeclareMathOperator*{\argmin}{argmin}
\DeclareMathOperator{\rank}{rank}
\DeclareMathOperator{\Lap}{Lap}
\DeclareMathOperator{\Lapm}{\Lap^m}

\DeclareMathOperator{\pdf}{pdf}

\newcommand*\mcup{\mathbin{\mathpalette\mcupinn\relax}}
\newcommand*\mcupinn[2]{\vcenter{\hbox{$\mathsurround=0pt
			\ifx\displaystyle#1\textstyle\else#1\fi\bigcup$}}}

\title{Differentially Private Linear Regression over Fully Decentralized Datasets}

\author{%
  Yang~Liu\\
  Tencent Cloud Product Department\\
  Tencent\\
  Shenzhen 518057, China \\
  \texttt{clarkieliu@tencent.com} \\
   \And
   Xiong~Zhang \\
   Tencent Cloud Product Department \\
   Tencent \\
   Shenzhen 518057, China \\
   \texttt{farleyzhang@tencent.com} \\
   \And
   Shuqin~Qin \\
   Tencent Cloud Product Department \\
	Tencent \\
	Shenzhen 518057, China \\
   \texttt{sookieqin@tencent.com} \\
   \And
   Xiaoping~Lei \\
   Tencent Cloud Product Department \\
	Tencent \\
	Shenzhen 518057, China \\
   \texttt{edenlei@tencent.com} \\
}

\newtheorem{theorem}{Theorem}

\newtheorem{definition}{Definition}
\newtheorem{assumption}{Assumption}

\begin{document}

\maketitle

\begin{abstract}
This paper presents a differentially private algorithm for linear regression learning in a decentralized fashion. Under this algorithm, privacy budget is theoretically derived, in addition to that the solution error is shown to be bounded by $\mO(t)$ for $\mO(\frac{1}{t})$ descent step size and $\mO(\exp(t^{1-e}))$ for $\mO(\frac{1}{t^e})$ descent step size.
\end{abstract}

\section{Introduction}

In recent years, optimization and learning among fully decentralized parties are drawing much attention \cite{nedic2009distributed,nedic2010constrained,boyd2011distributed}. However, privacy concerns are not taken into account in much of the work. Although \cite{huang2015differentially} presents a private distributed convex optimizer by incorporating the famous notion of differential privacy \cite{dwork2011differential}, too strong boundedness assumptions on the objectives must hold. In this paper, we specify the objective as the famous least squares, and provide a differentially private decentralized solver, as well as privacy and accuracy results with relaxed assumptions.

%\noindent{\bf Notations.}
%%We use $\vb[i]$ to represent the $i$-th component of the vector $\vb$.
%Let $\|\cdot\|_p$ denote the $L^p$-norm of a vector or matrix, and $\|\cdot\|$ denote the $L^2$-norm by default. Let $\Lap(v)$ with $v>0$ denote the Laplace distribution with zero mean and $2v^2$ variance, whose probability density function is given by $\pdf(x;v)$. Correspondingly, $\Lapm(v)$ is an $m$-dimensional vector of i.i.d. Laplace random variables according to $\Lap(v)$. We denote the Kronecker product by $\otimes$. Let $\diag(\Ab_1,\dots,\Ab_k)$ denote a block diagonal matrix with its diagonal blocks being $\Ab_1,\dots,\Ab_k$ for top left to bottom right. We denote $\mO(\cdot)$ as the asymptotic upper bound of a function.

\section{Problem Definition}

\subsection{Decentralized Datasets over Networks}
Let $\mV=\{1,\dots,k\}$ represent a group of decentralized parties that aim to participate in a global computational task. As a setup of this paper, the parties in $\mV$, termed as nodes, are peer-to-peer interconnected to locally establish two-way communication, described by edges in a set of unordered pair of nodes $\mE=\{\{i,j\}:i,j\textnormal{ are connected},i,j\in\mV\}$. Based on the edge set $\mE$, one can define the neighbor set of node $i$ as $\mN_i=\{j:\{i,j\}\in\mE\}\mcup\{i\}$. Over such a network $\mG=(\mV,\mE)$, which is assumed to be connected throughout this paper, nodes $i\in\mV$ hold mutually exclusive and homogeneous datasets $\mcD_i\in\R^{n_i\times m}\times\R^{n_i}$, respectively, including the design matrix $\Xb_i\in\R^{n_i\times m}$ and the label vector $\yb_i\in\R^{n_i}$. 	One of the foundational assumptions of this paper is that $\mcD_i$ is seen as privacy by each node $i$.

\subsection{Existing Decentralized Linear Regression Algorithm}

Linear regression is a common model that arises in various disciplines. Consider a design matrix $\Xb\in\R^{n\times m}$ and a label vector $\yb\in\R^n$. Then the learning goal of linear regression is to solve the following least-squares problem:
\begin{equation}\label{eq:linear_regression}
\min_{\betab\in\R^m}\qquad \frac{1}{2}\|\Xb\betab-\yb\|^2.
\end{equation}
It is well-known that (\ref{eq:linear_regression}) yields a unique optimal estimate $\betab^\ast=(\Xb^\top\Xb)^{-1}\Xb^\top\yb$ if $\Xb$ has full column rank.
By letting $n = \sum\limits_{i=1}^k n_i$, $\Xb=[\Xb_1^\top\ \dots\ \Xb_k^\top]^\top$ and $\yb=[\yb_1^\top\ \dots\ \yb_k^\top]^\top$,
%\begin{align}
%\Xb=\begin{bmatrix}
%\Xb_1\\
%\vdots\\
%\Xb_k
%\end{bmatrix},\
%\yb=\begin{bmatrix}
%\yb_1\\
%\vdots\\
%\yb_k
%\end{bmatrix},\notag
%\end{align}
we finally obtain a decentralized linear regression modelling task (\ref{eq:linear_regression}) over network $\mG$. A fully decentralized algorithm for solving (\ref{eq:linear_regression}) is described by the following dynamics \cite{nedic2010constrained}:
\begin{equation}\label{eq:nedic}
\betab_i(t+1) = \sum\limits_{j\in\mN_i} w_{ij}\betab_j(t) - \alpha(t)\nabla L_i(\betab_i(t)), 
\end{equation}
where $t=0,1,2,\dots$ is the discretized time, $\betab_i(t)$ is node $i$'s current estimate towards the global model, edge weight $w_{ij}>0$ is defined over $j\in\mN_i$ satisfying $w_{ij}=w_{ji}$ and $\sum\limits_{j\in\mN_i}w_{ij}=1$ for all $i\in\mV$, $\alpha:\Z^{\ge0}\to\R^+$ is the step size, and
$L_i(\betab)=\frac{1}{2}\|\Xb_i\betab-\yb_i\|^2$.
It was proved that if $\sum\limits_{t=0}^\infty\alpha(t)=\infty$ and $\lim\limits_{t\to\infty}\alpha(t)=0$, then $\lim\limits_{t\to\infty} \betab_i(t)=\betab^\ast$ for all $i\in\mV$ \cite{liu2018network}. Typical selections of $\alpha(t)$ include $\alpha(t)=\frac{c}{(t+d)^e}$ with $c,d>0$ and $0<e\le1$. Evidently, the contents shared among nodes are $\{\betab_i(t)\}_{i\in\mV,t\in\Z^{\ge0}}$, which contain the information of $\nabla L_i$ and thereby $\mcD_i$. When confronted with global adversaries capable of observing the communication contents, the algorithm (\ref{eq:nedic}) leads to undesirable privacy disclosure. Therefore, a privacy-preserving version of (\ref{eq:nedic}) is demanded.

\section{Main Results}

In this section, we propose a privacy-preserving version of (\ref{eq:nedic}), and provide corresponding differential privacy and accuracy analysis. To facilitate the presentation of our algorithm, we first introduce the following assumption.
\begin{assumption}\label{ass:Omega}
	All nodes of the network $\mG$ knows that the optimal estimate $\betab^\ast\in\R^m$ falls into a compact and convex set $\Omega\subset\R^m$ with $B_\Omega=\sup\limits_{\betab\in\Omega} \|\betab\|$.
\end{assumption}
Note that Assumption \ref{ass:Omega} is reasonable in the sense that heuristic approaches can be applied to find $\Omega$. For example, if $\rank(\Xb_i)=m$, each node $i$ can present a convex set $\Omega_i\subset\R^m$ containing its local optimal estimate $\betab^\ast_i=\argmin\limits_{\betab\in\R^m}L_i(\betab)$, and $\Omega$ can be set as a convex hull of $\bigcup\limits_{i\in\mV}\Omega_i$. Such methods are out of scope, and thereby not comprehensively investigated in this paper.

\subsection{Privacy-Preserving Algorithm}

Define $\PPo(\betab)=\inf\limits_{\betab^\prime\in\Omega}\|\betab-\betab^\prime\|$ as the projection onto $\Omega$. Inspired by (\ref{eq:nedic}), we provide the following privacy-preserving linear regression algorithm that terminates in finite time $T\ge1$.

\begin{algorithm}[H]
	\begin{algorithmic}[1]
		\STATE Set $t\gets0$ and initialize $\betab_i(0)$ for all $i\in\mV$.
		\STATE Each node $i$ draws $\omegab_i(t)\in\R^m$ from the distribution $\Lapm(v(t))$ satisfying $\lim\limits_{t\to\infty}v(t)=0$.
		\STATE Each node $i$ computes and propagates $\betaf_i(t)\gets\betab_i(t)+\omega(t)$ to its neighbors $j\in\mN_i$.
		\STATE Each node $i$ computes the projected state $\betas_i(t)\gets\PPo(\betaf_i(t))$.
		\STATE Each node $i$ updates its state by $\betab_i(t+1) \gets \sum\limits_{j\in\mN_i}w_{ij}\betas_j(t)-\alpha(t)\nabla L_i(\betas_i(t))$.
		\STATE Set $t\gets t+1$. Algorithm terminates if $t=T$, otherwise go to Step 2.
	\end{algorithmic}
	\caption{$T$-step Privacy-Preserving Linear Regression}
	\label{alg:dplr}
\end{algorithm}
As can be noted, under Algorithm \ref{alg:dplr} each node injects Laplace random noise before true estimate propagation. After receiving the slightly distorted estimate, each node projects it onto the convex set containing the optimum to avoid the divergence of learning process.

\subsection{Differential Privacy}

Now we analyze the differential privacy of Algorithm \ref{alg:dplr}. Relevant notions based on \cite{dwork2011differential} are provided in the following.
\begin{definition}
	Consider two network datasets $\mcD=(\Xb,\yb)$ and $\mcD^\prime=(\Xb^\prime,\yb^\prime)$ in $\R^{n\times m}\times\R^{n}$ with $n=\sum\limits_{i=1}^k n_i$. Then $\mcD$ and $\mcD^\prime$ are said to be $(\delta_X,\delta_y)$-adjacent if there exists $i\in\{1,\dots,k\}$ such that (i) $\|\Xb_i\|,\|\Xb_i^\prime\|\le\delta_X$ and $\|\yb_i\|,\|\yb_i^\prime\|\le\delta_y$; (ii) $\Xb_j=\Xb_j^\prime$ and $\yb_j=\yb_j^\prime$ for all $j\neq i$.
%	\begin{enumerate}[(i)]
%		\item $\|\Xb_i\|,\|\Xb_i^\prime\|\le\delta_X$ and $\|\yb_i\|,\|\yb_i^\prime\|\le\delta_y$;
%		\item $\Xb_j=\Xb_j^\prime$ and $\yb_j=\yb_j^\prime$ for all $j\neq i$.
%	\end{enumerate}
\end{definition}
	Clearly, the adversaries against Algorithm \ref{alg:dplr} observe all communication contents among nodes $\{\betaf_i(t)\}_{i\in\mV,t=0,\dots,T-1}$, based on which they aim to infer the privacy $\mcD$. Such an adversarial relation can be intrinsically described by a mapping $\MM_T:\R^{n\times m}\times\R^n\times\R^{km}\to\R^{kmT}$ with
	$$\MM_T(\mcD,\{\betaf_i(0)\}_{i\in\mV})= \{\betaf_i(t)\}_{i\in\mV,t=0,\dots,T-1}.$$
	Then the following definition is provided on the differential privacy of Algorithm \ref{alg:dplr}.
\begin{definition}
Algorithm \ref{alg:dplr} in $T$-step preserves $\epsilon$-differential privacy under $(\delta_X,\delta_y)$-adjacency if for all $\mR\subset\R^{kmT}$ and for all $\{\betaf_i(0)\}_{i\in\mV}\in\R^{km}$, there holds
$$\Pr(\MM_T(\mcD,\{\betaf_i(0)\}_{i\in\mV})\in\mR)\le e^\epsilon \Pr(\MM_T(\mcD^\prime,\{\betaf_i(0)\}_{i\in\mV})\in\mR)$$
for all $(\delta_X,\delta_y)$-adjacent network datasets $\mcD,\mcD^\prime\in\R^{n\times m}\times\R^n$.
\end{definition}
For Algorithm \ref{alg:dplr}, we provide the following theorem.
\begin{theorem}\label{thm:differential_privacy}
Let Assumption \ref{ass:Omega} hold. Then there exists finite $\epsilon>0$ such that Algorithm \ref{alg:dplr} in $T$-step preserves $\epsilon$-differential privacy under $(\delta_X,\delta_y)$-adjacency as $T$ goes to infinity if $\big\{\frac{\alpha(t)}{v(t+1)}\big\}_{t=0}^\infty$ is summable. In particular, if $\alpha(t)=\frac{c_\alpha}{(t+d_\alpha)^{e_\alpha}}$ and $v(t)=\frac{c_v}{(t+d_v)^{e_v}}$ with $c_\alpha,e_\alpha,c_v,e_v>0$ and $1<d_v+1\le d_\alpha$, then Algorithm \ref{alg:dplr} in $T$-step preserves
$$4\delta_Xc_\alpha c_v^{-1}T\sqrt{mn_M}(\delta_XB_\Omega\sqrt{km}+\delta_y)$$
-differential privacy with $n_M=\max\{n_i:i\in\mV\}$.
\end{theorem}
\begin{proof}
We will use the compact notation
$
\betab(t)=[
\betab_1(t)^\top\ \dots\ \betab_k(t)^\top
]^\top\in\R^{km}
$
for $\betab_i(t)$, and the same form will also appear for $\betaf_i(t)$ and $\betas_i(t)$, whose introduction will be omitted. The underlying dynamics of Algorithm \ref{alg:dplr} can be written as
\begin{align}
\betaf(t+1) = (\Wb\otimes\Ib_m)\PPo^\ast(\betaf(t)) - \alpha(t)G(\PPo^\ast(\betaf(t))) + \omegab(t+1),\label{eq:dynamics}
\end{align}
where the $ij$--th element of $\Wb\in\R^{k\times k}$ equals $w_{ij}$ if $j\in\mN_i$ and zero otherwise,
$
\PPo^\ast(\betaf(t)) =
[
\PPo(\betaf_1(t))^\top\ \cdots\ \PPo(\betaf_k(t))^\top
]^\top$, and $
G(\PPo^\ast(\betaf(t))) =
[
\nabla L_1(\PPo(\betaf_1(t)))^\top\ \cdots\ \nabla L_k(\PPo(\betaf_k(t)))^\top
]^\top= \tX\PPo^\ast(\betaf(t))-\ty
$
%\begin{align}
%\PPo^\ast(\betaf(t)) &=
%\begin{bmatrix}
%\PPo(\betaf_1(t))^\top & \cdots & \PPo(\betaf_k(t))^\top
%\end{bmatrix}^\top,\notag\\
%G(\PPo^\ast(\betaf(t))) &=
%\begin{bmatrix}
%\nabla L_1(\PPo(\betaf_1(t)))^\top & \cdots & \nabla L_k(\PPo(\betaf_k(t)))^\top
%\end{bmatrix}^\top= \tX\PPo^\ast(\betaf(t))-\ty\label{eq:dp1}
%\end{align}
with $\tX=\diag(\Xb_1^\top\Xb_1,\dots,\Xb_k^\top\Xb_k)$ and $\ty=[\yb_1^\top\Xb_1\ \dots\ \yb_k^\top\Xb_k]^\top$. Define $\MM^{(t)}(\mcD,\betaf(t))=\betaf(t+1)$ such that $\MM_T(\{\betaf_i(0)\}_{i\in\mV})=\{\MM^{(\tau)}\circ\dots\circ\MM^{(0)}:\tau=0,\dots,T-1\}$ when omitting $\mcD$. Then for any $\mcD,\mcD^\prime$ differing at node $i^\ast$'s dataset w.l.o.g., there hold for all $t\ge0$ based on (\ref{eq:dynamics})
\begin{align}
&\quad\frac{\Pr(\MM^{(t)}(\mcD,\betaf(t))=\betaf(t+1))}{\Pr(\MM^{(t)}(\mcD^\prime,\betaf(t))=\betaf(t+1))}\notag\\
&\overset{\rm a)}{=} \frac{\pdf(\betaf(t+1)-(\Wb\otimes\Ib_m)\PPo^\ast(\betaf(t)) + \alpha(t)G(\PPo^\ast(\betaf(t)));v(t+1))}{\pdf(\betaf(t+1)-(\Wb\otimes\Ib_m)\PPo^\ast(\betaf(t)) + \alpha(t)G^\prime(\PPo^\ast(\betaf(t)));v(t+1))}\notag\\
%&\overset{\rm b)}{\le} \exp\big(\alpha(t)v^{-1}(t+1)\|G(\PPo^\ast(\betaf(t)))-G^\prime(\PPo^\ast(\betaf(t)))\|_1\big)\notag\\
&\overset{\rm b)}{\le}\exp\big(\alpha(t)v^{-1}(t+1)(\|\tX-\tX^\prime\|_1\|\PPo^\ast(\betaf(t))\|_1+\|\ty-\ty^\prime\|_1)\big)\notag\\
&\le\exp\big(\alpha(t)v^{-1}(t+1)(\|\Xb_{i^\ast}^\top\Xb_{i^\ast}-\Xb_{i^\ast}^{\prime\top}\Xb_{i^\ast}^\prime\|_1\|\PPo^\ast(\betaf(t))\|_1+\|\Xb_{i^\ast}^\top\yb_{i^\ast}-\Xb_{i^\ast}^{\prime\top}\yb_{i^\ast}^\prime\|_1)\big),\label{eq:dp2}
\end{align}
where a) is from the Laplace distribution and b) is an application of norm inequalities. Based on norm inequalities and equivalence \cite{horn2012matrix}, one has
\begin{equation}\label{eq:dp3}
\begin{aligned}
\|\Xb_{i^\ast}^\top\Xb_{i^\ast}-\Xb_{i^\ast}^{\prime\top}\Xb_{i^\ast}^\prime\|_1 &= \bigg\|
\begin{bmatrix}
\Xb_{i^\ast}^\top & -\Xb_{i^\ast}^{\prime\top}
\end{bmatrix}
\begin{bmatrix}
\Xb_{i^\ast}\\
\Xb_{i^\ast}^\prime
\end{bmatrix}
\bigg\|_1 \le (\|\Xb_{i^\ast}^\top\|_1 +\|\Xb_{i^\ast}^{\prime\top}\|_1)(\|\Xb_{i^\ast}\|_1+\|\Xb_{i^\ast}^\prime\|_1)\\
&\le \sqrt{mn_{i^\ast}}(\|\Xb_{i^\ast}\| +\|\Xb_{i^\ast}^{\prime}\|)^2 \le 4\delta_X^2\sqrt{mn_M}.
\end{aligned}
\end{equation}
Similarly, we have
\begin{equation}\label{eq:dp4}
\|\Xb_{i^\ast}^\top\yb_{i^\ast}-\Xb_{i^\ast}^{\prime\top}\yb_{i^\ast}^\prime\|_1 \le 4\delta_X\delta_y\sqrt{mn_M}.
\end{equation}
According to (\ref{eq:dp2}), (\ref{eq:dp3}) and (\ref{eq:dp4})
\begin{align}
\frac{\Pr(\MM^{(t)}(\mcD,\betaf(t))=\betaf(t+1))}{\Pr(\MM^{(t)}(\mcD^\prime,\betaf(t))=\betaf(t+1))}\le\exp\big(4\delta_X\sqrt{mn_M}(\delta_XB_\Omega\sqrt{km}+\delta_y)\alpha(t)v^{-1}(t+1)\big).\label{eq:dp5}
\end{align}
Based on (\ref{eq:dp5}) and the composition property \cite{mcsherry2009privacy}, this proof is completed.
\end{proof}

\subsection{Accuracy Analysis}
%A theoretic investigation on the accuracy is provided below.

\begin{theorem}\label{thm:accuracy}
Let Assumption \ref{ass:Omega} hold. Suppose $\alpha(t)=\mO(\frac{1}{t^{e_\alpha}})$ with $0<e_\alpha\le1$ and $v(t)=\mO(\frac{1}{t^{e_v}})$ with $e_v>0$. Then under Algorithm \ref{alg:dplr}, there holds
\begin{equation}\notag
\sum\limits_{i\in\mV}\E\|\betaf_i(t)-\betab^\ast\|=\left\{
\begin{aligned}
& \mO(t)&\textnormal{ if }e_\alpha=1;\\
& \mO(\exp(t^{1-e_\alpha}))&\textnormal{ otherwise.}
\end{aligned}
\right.
\end{equation}
\end{theorem}
\begin{proof}
We will continue to use the notations in the proof of Theorem \ref{thm:differential_privacy}. Define $\eb(t)=\betaf(t)-\1\otimes\betab^\ast$. By subtracting $\1\otimes\betab^\ast$ on both sides of (\ref{eq:dynamics}), one has
\begin{align}
\eb(t+1)=(\Wb\otimes\Ib-\alpha(t)\tX)\eb(t)+\alpha(t)(\ty-\tX(\1\otimes\betab^\ast))+\omegab(t+1).\label{eq:acc1}
\end{align}
Then it follows (\ref{eq:acc1})
\begin{align}
\|\eb(t+1)\|^2 &\le \eb(t)^\top(\Wb\otimes\Ib-\alpha(t)\tX)^2\eb(t)+\alpha^2(t)\|\ty-\tX(\1\otimes\betab^\ast)\|^2+\|\omegab(t+1)\|^2\notag\\
&\quad+\alpha(t)\|\ty-\tX(\1\otimes\betab^\ast)\|\|\Wb\otimes\Ib-\alpha(t)\tX\|\|\eb(t)\| + g(\omegab(t+1)),\label{eq:acc2}
\end{align}
where $g:\R^{km}\to\R^{km}$ is linear. Due to the nonnegativity and irreducibility of $\Wb$ \cite{horn2012matrix}, there holds $-1\le\|\Wb\|<1$, and thereby $\|\Wb\otimes\Ib-\alpha(t)\tX\|\le1+\alpha(t)$. Then by (\ref{eq:acc2})
\begin{align}
\E\|\eb(t+1)\|^2 &= \mO\big((1+\alpha(t))^2\E\|\eb(t)\|^2+\alpha(t)(1+\alpha(t))\|\eb(t)\|+\alpha^2(t)+v^2(t)\big)\notag\\
&=\mO\big(\big((1+\alpha(t))\E\|\eb(t)\|+\alpha(t)\big)^2+v^2(t)\big),\notag
\end{align}
which further leads to
\begin{align}
\E\|\eb(t+1)\| &= \mO\big((1+\alpha(t))\E\|\eb(t)\|+\alpha(t)+v(t)\big)\notag\\
&=\mO\bigg(\prod\limits_{\tau=0}^t(1+\alpha(\tau))+\sum\limits_{\tau=0}^t(\alpha(\tau)+v(\tau))\prod\limits_{\kappa=\tau+1}^{t}(1+\alpha(\kappa))\bigg)\notag\\
&=\mO\bigg(\exp\big(\sum\limits_{\tau=0}^t\alpha(\tau)\big)+\sum\limits_{\tau=0}^t\big(\alpha(\tau)+v(\tau)\big)\exp\big(\sum\limits_{\kappa=\tau+1}^t\alpha(\kappa)\big)\bigg).\label{eq:acc3}
\end{align}
It is a fact
$
\sum\limits_{\tau=t^\prime}^{t} \frac{1}{\tau^e}=\mO\big(\int_{t^\prime-1}^{t}\frac{1}{\tau^e}\mathrm{d\tau}\big)
$
for all $t\ge t^\prime>1$. Based on (\ref{eq:acc3}), one has
\begin{equation}\label{eq:acc5}
\E\|\eb(t+1)\|=\left\{
\begin{aligned}
& \mO\bigg(t+t\sum\limits_{\tau=0}^t \frac{\alpha(\tau)+v(\tau)}{\tau}\bigg)&\textnormal{ if }e_\alpha=1;\\
& \mO\bigg(\exp(t^{1-e_\alpha})+\exp(t^{1-e_\alpha})\sum\limits_{\tau=0}^t \frac{\alpha(\tau)+v(\tau)}{\exp(\tau^{1-e_\alpha})}\bigg)&\textnormal{ otherwise.}
\end{aligned}
\right.
\end{equation}
Clearly, both $\sum\limits_{\tau=0}^\infty \frac{\alpha(\tau)+v(\tau)}{\tau}$ and $\sum\limits_{\tau=0}^\infty \frac{\alpha(\tau)+v(\tau)}{\exp(\tau^{1-e_\alpha})}$ are convergent, the proof is completed by (\ref{eq:acc5}).
\end{proof}

\section{Conclusions}

In this paper, a differentially private decentralized algorithm for linear regression was proposed. Not only a theoretic privacy budget was provided, but the precision was carefully investigated and shown to be bounded by $\mO(t)$ or $\mO(\exp(t^{1-e}))$. Future work includes the tradeoff analysis between efficiency and privacy, and the relaxation of the projection operation.

\bibliographystyle{unsrtnat}
\bibliography{main}

\begin{thebibliography}{8}
\providecommand{\natexlab}[1]{#1}
\providecommand{\url}[1]{\texttt{#1}}
\expandafter\ifx\csname urlstyle\endcsname\relax
  \providecommand{\doi}[1]{doi: #1}\else
  \providecommand{\doi}{doi: \begingroup \urlstyle{rm}\Url}\fi

\bibitem[Nedic and Ozdaglar(2009)]{nedic2009distributed}
Angelia Nedic and Asuman Ozdaglar.
\newblock Distributed subgradient methods for multi-agent optimization.
\newblock \emph{IEEE Transactions on Automatic Control}, 54\penalty0
  (1):\penalty0 48, 2009.

\bibitem[Nedic et~al.(2010)Nedic, Ozdaglar, and Parrilo]{nedic2010constrained}
Angelia Nedic, Asuman Ozdaglar, and Pablo~A Parrilo.
\newblock Constrained consensus and optimization in multi-agent networks.
\newblock \emph{IEEE Transactions on Automatic Control}, 55\penalty0
  (4):\penalty0 922--938, 2010.

\bibitem[Boyd et~al.(2011)Boyd, Parikh, Chu, Peleato, Eckstein,
  et~al.]{boyd2011distributed}
Stephen Boyd, Neal Parikh, Eric Chu, Borja Peleato, Jonathan Eckstein, et~al.
\newblock Distributed optimization and statistical learning via the alternating
  direction method of multipliers.
\newblock \emph{Foundations and Trends{\textregistered} in Machine learning},
  3\penalty0 (1):\penalty0 1--122, 2011.

\bibitem[Huang et~al.(2015)Huang, Mitra, and Vaidya]{huang2015differentially}
Zhenqi Huang, Sayan Mitra, and Nitin Vaidya.
\newblock Differentially private distributed optimization.
\newblock In \emph{Proceedings of the 2015 International Conference on
  Distributed Computing and Networking}, page~4. ACM, 2015.

\bibitem[Dwork(2011)]{dwork2011differential}
Cynthia Dwork.
\newblock Differential privacy.
\newblock \emph{Encyclopedia of Cryptography and Security}, pages 338--340,
  2011.

\bibitem[Liu et~al.(2018)Liu, Lou, Anderson, and Shi]{liu2018network}
Yang Liu, Youcheng Lou, Brian Anderson, and Guodong Shi.
\newblock Network flows that solve least squares for linear equations.
\newblock \emph{arXiv preprint arXiv:1808.04140}, 2018.

\bibitem[Horn and Johnson(2012)]{horn2012matrix}
Roger~A Horn and Charles~R Johnson.
\newblock \emph{Matrix {A}nalysis}.
\newblock Cambridge university press, 2012.

\bibitem[McSherry(2009)]{mcsherry2009privacy}
Frank~D McSherry.
\newblock Privacy integrated queries: an extensible platform for
  privacy-preserving data analysis.
\newblock In \emph{Proceedings of the 2009 ACM SIGMOD International Conference
  on Management of data}, pages 19--30. ACM, 2009.

\end{thebibliography}

\end{document}